\documentclass[journal,twoside,web]{ieeecolor}
\def\logoname{blank_logo}
\definecolor{subsectioncolor}{rgb}{0,0.541,0.855}
\usepackage[style=ieee]{biblatex}
\usepackage[dvipsnames]{xcolor}

\usepackage{./TLC}

\usepackage{circuitikz}
\usepackage{blox}
\usepackage{pgfplots}
\pgfplotsset{compat=1.16}
\usepgfplotslibrary{groupplots}

\pgfplotscreateplotcyclelist{colors}{
        {CornflowerBlue},
        {BurntOrange},
        {RubineRed},
        {Emerald},
        {SpringGreen},
        {Goldenrod}
}

\usepackage{mathtools}
\usepackage{gensymb}
\usepackage{amsmath}
\usepackage{amssymb}
\usepackage{siunitx}
\usepackage{algorithm}
\usepackage{algpseudocode}

\usepackage{epigraph}
\usepackage{shellesc}
\usepackage{flushend}
\title{Monotone one-port circuits}

\author{Thomas Chaffey$^{1}$ and Rodolphe Sepulchre$^{1,2}$
\thanks{The research leading to these results has received funding from the European
Research Council under the Advanced ERC Grant Agreement SpikyControl n. 101054323.}%
\thanks{$^{1}$University of Cambridge, Department of Engineering, Trumpington Street,
Cambridge CB2 1PZ, {\tt\small tlc37@cam.ac.uk}, {\tt\small r.sepulchre@eng.cam.ac.uk}.}
\thanks{$^{2}$KU Leuven, Department of Electrical Engineering (STADIUS), KasteelPark
Arenberg, 10, B-3001 Leuven, Belgium.}
}
\begin{document}
\maketitle

\begin{abstract}
        Maximal monotonicity is explored as a generalization of the linear theory of
        passivity, aiming at an algorithmic input/output analysis of physical models.  The theory is developed for  maximal
        monotone one-port circuits, formed by the series and parallel interconnection
        of basic elements.  These circuits generalize passive LTI transfer functions.
        Periodic input signals are shown to be mapped to periodic output signals, and
        these input-output behaviors can be efficiently computed 
        using a maximal monotone splitting algorithm,
        which decomposes the computation according to the circuit topology. A new
        splitting algorithm is presented, which applies to any monotone one-port circuit
        defined as a port interconnection of monotone elements.
\end{abstract}

\section{Introduction}

Passivity is the backbone of linear circuit theory. As a system theoretic concept,
it provides a fundamental bridge between physics and computation, well
beyond electrical circuits.  The concept of passivity first arose in the study of
interconnections of basic circuit elements \autocite{Foster1924, Cauer1926, Brune1931}, and passive
transfer functions are precisely those that can be realized as port interconnections
of passive elements \autocite{Bott1949}. The KYP lemma provides an
algorithmic framework for the analysis of passive circuits by convex optimization
\autocite{Yakubovich1962, Kalman1963, Popov1964}.
The circuit concept of passivity has generated amongst the most important developments
of control theory over the last several decades, including dissipativity theory
\autocite{Willems1972, Willems1972a}, nonlinear
passivity theory \autocite{Hill1976, Hill1980, Moylan2014}, and passivity based
control \autocite{vanderSchaft2017, Jayawardhana2006, Ortega1998, Sepulchre1997}.

The strong physical and computational properties of linear passive circuits fail to
generalize to the nonlinear world.  The mere existence and uniqueness
 of solutions can no longer be taken for granted, let alone an algorithm to compute
 them.
 Contrary to physical intuition, a nonlinear passive resistor may have regions of
 negative resistance \autocite{Chua1983}.  Negative resistance circuits behave quite
 differently from linear passive circuits, being the source of switches and
 oscillations \autocite{Miranda-Villatoro2022}.  
 This paper explores \emph{maximal monotonicity} as a generalization of linear
 passivity that retains both its physical and algorithmic significance in nonlinear
 circuits.
In the spirit of the early work of Foster \autocite{Foster1924}, Cauer
\autocite{Cauer1926}, Brune \autocite{Brune1931}, Bott and Duffin
\autocite{Bott1949}, and the ``tearing, zooming and linking'' methodology advocated
by Willems \autocite{Willems2007}, we consider a class of systems
formed by port interconnections of basic elements, possibly nonlinear. The primary message of this paper is that,
like passive LTI transfer functions, this class is both physical and computationally
tractable.

This proposal is classical, and indeed the property of maximal monotonicity first
arose in efforts to extend the
tractability of linear, time invariant, passive networks to networks
of nonlinear resistors.   The prototype of a maximal monotone
element was Duffin's
\emph{quasi-linear} resistor \cite{Duffin1946}, a nonlinear resistor with a
non-decreasing $i-v$ characteristic.  Other early forms of monotonicity are
found in the work of \textcite{Golomb1935}, \textcite{Zarantonello1960} and
the work of \textcite{Dolph1961} on ``dissipative'' linear mappings.  
Quasi-linearity was refined by Minty
\cite{Minty1960, Minty1961, Minty1961a} to produce the modern concept of
maximal monotonicity, in the context of an algorithm for solving networks of
nonlinear resistors.  Desoer and Wu \cite{Desoer1974} studied existence and
uniqueness of solutions to networks of nonlinear resistors, capacitors and
inductors defined by maximal monotone relations.

Following the influential paper of Rockafellar in 1976
\cite{Rockafellar1976}, maximal monotonicity has grown to become a
fundamental property in convex optimization \cite{Rockafellar1997, Ryu2016,
Ryu2022, Parikh2013, Bertsekas2011, Combettes2011}, forming the basis of a large body of
work on tractable first-order methods for large scale and nonsmooth
optimization problems, which have seen a surge of interest in the last
decade.  The principle notion is that of \emph{splitting}: for operators which can be
separated into sums of monotone operators, a \emph{splitting algorithm} can be
applied which separates computation for each element of the sum, allowing computation
to be distributed across multiple devices.  This is the basis for many popular
first-order methods, including ADMM and proximal gradient.  For a comprehensive
bibliography, we refer the reader to the literature review of
\autocite[Ch. 2]{Ryu2022}.
It seems, however, that the physical significance of maximal monotonicity has been 
somewhat forgotten.  In this paper, we revisit the
study of nonlinear circuits in light of modern developments in
the theory of splitting algorithms.

Maximal monotonicity also plays a fundamental role in a long line of research on
analysis and simulation of \emph{state space} systems interconnected with non-smooth and set-valued
components, recently surveyed by Brogliato and Tanwani \autocite{Brogliato2020}.  The first connection
between maximal monotone operators and passive linear systems appears in this area, in the work of
\textcite{Brogliato2004}. This work
inspired a line of research on Lur'e systems consisting of a passive LTI system in
feedback with a nonsmooth maximal monotone operator (see, for instance, \autocite{Brogliato2009, Brogliato2011a,
Brogliato2013, Camlibel2016, Adly2017, Brogliato2020}).   Maximal monotone
differential inclusions were first studied by Br\'ezis, who, in particular, proved
existence of periodic solutions when such differential inclusions are periodically
forced by a locally integrable input \autocite[Ch. 3, $\S$6]{Brezis1973}.
Existence and uniqueness of solutions to such differential inclusions have since been
studied extensively \autocite{Camlibel2016, Camlibel2022b}.  Solutions can be constructed
using the classical time-stepping algorithm surveyed in \autocite[$\S$5.2]{Brogliato2020}.  A number of other time domain simulation algorithms have
been developed for Lur'e systems with maximal monotone nonlinearities in the feedback
\autocite{Acary2008, Brogliato2016}.  Two algorithms for computing the periodic response of such
Lur'e systems are given by Heemels \emph{et al.}
\cite{Heemels2017}, which both involve iteratively computing the resolvent of a
differential inclusion.
Lur'e systems with maximal monotone nonlinearities may also be modelled as linear
complementarity systems, and specialized methods for computing steady state
oscillations in such systems have been developed by Ianelli \emph{et al.}
\autocite{Iannelli2011} and Meingast \emph{et al.} \autocite{Meingast2014}.
Other methods for computing periodic responses of nonlinear systems
are either approximate and limited in their
applicability, as in harmonic analysis
\cite{Feldmann1996, Blagquiere1966, Krylov1947, Slotine1991} or involve performing a
transient simulation and waiting for convergence \cite{Aprille1972, Cellier2006}.

In this paper, we also study the periodic response of maximal monotone circuits.
However, our approach in grounded in  input-output rather than state-space descriptions. We
treat circuits as
a physical interconnection of basic components, each of which is required to be a
maximal monotone operator \emph{on the space of $T$-periodic $i-v$ trajectories}.
  Rather
than performing an integration forwards in time, we pose the periodic response of the
circuit as a zero of an operator on the
space of periodic trajectories, and draw on the splitting algorithms of convex
optimization.  This
method is reminiscent of 
frequency response analysis of an LTI transfer function.  Standard splitting algorithms allow
the computation to be organized using the interconnection structure in the case of
purely series or purely parallel circuits, an observation first made in the
conference version of this paper \autocite{Chaffey2021}.  Here we further develop
this idea and show that the circuit topology of the circuit can be put in direct
correspondence with its algorithmic solution. We introduce a splitting algorithm suited to arbitrary series/parallel
circuits, which generalizes existing splitting algorithms, which find zeros of
sums of maximal monotone operators, to nested sums
and inverses of maximal monotone operators.

The remainder of this paper is structured in three parts.
The first part of this paper describes, in general terms, the class of systems formed
from the series/parallel interconnection of maximal monotone one-port elements --
referred to throughout the paper as \emph{monotone one-port circuits}. In
Section~\ref{sec:illustrative_example},
we motivate the study of this class by contrasting it with interconnections of
passive elements.  While monotone one-port circuits retain the fundamental properties
of interconnections of LTI resistors, these properties are lost for passive nonlinear elements.  In
Section~\ref{sec:elements}, we review the basic theory of
        maximal monotone operators. In Section~\ref{sec:1ports}, we revisit port
        interconnections of maximal monotone elements, and formalize the class of
        systems studied in this paper.

        The second part of this paper develops an algorithmic framework for studying
        the periodic response of monotone one-port circuits.
          In section~\ref{sec:computation}, we first show that off-the-shelf
          optimization algorithms can be used to compute the periodic response of
          circuits composed of purely parallel or purely series interconnections.  We
          then introduce a new splitting algorithm (Algorithm~\ref{alg:nested}). 

        The final part of the paper studies particular, physical classes of monotone
        one-port circuits.  
        Section~\ref{sec:RLC} applies the theory to nonlinear RLC circuits, and 
        gives two detailed computational examples, including a
        large-scale circuit consisting of 300,000 elements\footnote{The code for the
        examples of this paper is available at
\url{github.com/ThomasChaffey/monotone-one-port-circuits}.}.
        Section~\ref{sec:conductance} applies the theory to memristive systems, using
        the specific example of a neuronal potassium conductance.

\section{Motivating example}\label{sec:illustrative_example}

We begin with a simple example, which motivates the developments of this paper:
the series interconnection of two resistors (Figure~\ref{fig:series-resistors}).

\begin{figure}[h!]
\centering
        \includegraphics{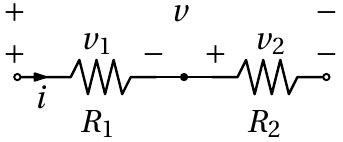}
        \caption{Series interconnection of two resistors.}
        \label{fig:series-resistors}
\end{figure}

Consider first the linear, time invariant case, $v_j = R_j i_j$, where each $R_j > 0$.  The series
interconnection maps the applied voltage $v$ to the port current $i$ by the relation
$i = v/(R_1 + R_2)$.  The interconnection has the property that
for any $T$-periodic applied voltage, there exists a unique $T$-periodic port
current.  This property seems natural for a circuit of passive elements.

The fundamentals of the circuit remain unchanged if the resistor $R_1$ is replaced
by a linear, time-invariant and passive transfer function $R_1(s)$, which obeys, by the positive real lemma, $\Re R_1(j\omega) \geq 0$
for all $\omega$.  
The resulting transfer function from current to voltage is
given by $1/(R_1(s) + R_2)$; given a $T$-periodic voltage $v$, the unique corresponding
current can be found by taking the Fourier series of $v$ and multiplying each
coefficient $\hat v_n$ by $1/(R_1(jnw\pi/T) + R_2)$.  Since $\Re R_1(j\omega) \geq
0$ and $R_2 > 0$, this complex number always has finite magnitude.

If we replace $R_1$ by a nonlinear, but passive resistor, however,
a $T$-periodic voltage no longer guarantees a $T$-periodic current.   
A passive resistor can have regions of negative slope in its
$i-v$ curve (\textcite{Chua1983} give a catalogue of physical
examples).  The voltage to current map of such a resistor is then multi-valued, and
it is not guaranteed that the current corresponding to a $T$-periodic voltage is
uniquely defined, nor $T$-periodic.  

If, however, we replace $R_1$ with a monotone nonlinear resistor,
the fundamentals of the LTI case remain unchanged.  Monotonicity of a resistor means
its $i-v$ curve is nondecreasing. The interconnection with $R_2 > 0$ means the
$i-v$ curve of the circuit is strictly increasing,  so invertibility of the
interconnection is retained, as illustrated in
Figure~\ref{fig:monotone-passive}.  This idea will be formalized for monotone RLC
circuits in Theorem~\ref{thm:periodic}.

\begin{figure}[h]
        \centering
        \includegraphics{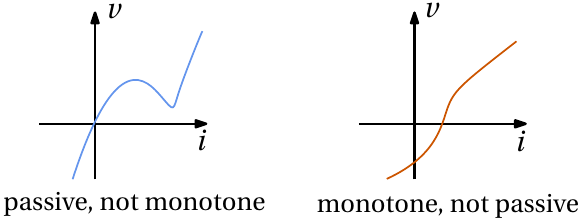}
        \caption{The $i$-$v$ curves of a passive and a monotone resistor.}%
        \label{fig:monotone-passive}
\end{figure}
\section{Maximal monotone relations}
\label{sec:elements}

We begin by introducing some mathematical preliminaries.

\subsection{Operators and relations}

\begin{definition}
        An operator on a space $X$ is a set-valued mapping $S:X \rightrightarrows X$.  The
        \emph{relation} or \emph{graph} of an operator $S$ is a subset
        $\operatorname{gra} S \subseteq X
        \times X$ of input-output pairs $(u, y)$.
\end{definition}

Throughout this paper, we will treat an operator and its relation as equivalent, and
often merely write $S$ when we are strictly referring to $\operatorname{gra} S$.
The operator notation $y \in S(u)$ is equivalent to the relational notation $(u, y)
\in \operatorname{gra} S$.

The usual operations on functions can be extended to set-valued operators: 
\begin{IEEEeqnarray*}{rCl}
        S^{-1} &=& \{ (y, u) \; | \; y \in S(u) \}\\
        S + R &=& \{ (x, y + z) \; | \; (x, y) \in S,\; (x, z) \in R \}\\
        SR &=& \{ (x, z) \; | \; \exists\, y \text{ s.t. } (x, y) \in R,\; (y, z) \in S \}.
\end{IEEEeqnarray*}
Note that the relational inverse $S^{-1}$ always exists, but in general, $S S^{-1}
\neq I$, where $I$ is the identity relation $\{(x, x)\;|\;x \in X\}$.

\subsection{Signal spaces}

Throughout this paper, we let $\mathcal{H}$ be a
Hilbert space with inner product $\ip{\cdot}{\cdot}$ and induced norm
$\norm{x} = \sqrt{\ip{x}{x}}$.  
The particular Hilbert spaces we consider are spaces of periodic signals, described
by a single period. A trajectory $w(t)$ is said to be
$T$-periodic if $w(t) = w(t + T)$ for all $t$.

Let $L_{2,T}$ denote the space of signals $u: [0, T] \to \R^n$ which are square integrable, that is,
\begin{IEEEeqnarray*}{rCl}
        \int_{0}^{T} u\tran(t) u(t) \dd{t} < \infty.
\end{IEEEeqnarray*}
This is a Hilbert space with inner product
\begin{IEEEeqnarray*}{rCl}
        \ip{u}{y} \coloneqq \int_{0}^{T} u\tran(t)y(t) \dd{t}
\end{IEEEeqnarray*}
and induced norm $\norm{u} \coloneqq \sqrt{\ip{u}{u}}$.
The discrete-time counterpart of $L_{2, T}$ is denoted by $l_{2, T}$, the space of
length $T$ sequences which are square summable:
\begin{IEEEeqnarray*}{rCl}
        \Sigma_{t = 0}^T u\tran(t)u(t) < \infty.
\end{IEEEeqnarray*}
Again,
this is a Hilbert space with inner product
\begin{IEEEeqnarray*}{rCl}
        \ip{u}{y} \coloneqq \Sigma_{t = 0}^{T} u\tran(t)y(t)
\end{IEEEeqnarray*}
and induced norm $\norm{u} \coloneqq \sqrt{\ip{u}{u}}$.

$L_{2, [0, \infty)}$ and $L_{2, (-\infty, \infty)}$ are defined analogously to $L_{2,
T}$, but with time intervals $[0, \infty)$ and $(-\infty, \infty)$, respectively.

\subsection{Maximal monotonicity}

The property of \emph{monotonicity} connects the physical property of energy
dissipation in a device to algorithmic analysis methods.  

Monotonicity on $\mathcal{H}$ is defined as follows.

\begin{definition}
        An operator $S: \mathcal{H}\rightrightarrows\mathcal{H}$ is called 
        \emph{monotone} if
\begin{IEEEeqnarray*}{rCl}
        \langle u_1 - u_2 | y_1 - y_2 \rangle \geq 0
\end{IEEEeqnarray*}
for any $(u_1, y_1), (u_2, y_2) \in \operatorname{gra} S$.
A monotone operator is called \emph{maximal} if its relation is not properly
contained in the relation of any
other monotone operator.
\end{definition}

By way of example, an operator $S: \R \rightrightarrows \R$ is monotone if its graph is
non-decreasing, and maximal if its graph has no endpoints.
Note that this definition refers to monotonicity in the operator theoretic sense, and
this is distinct from the notion of monotonicity in the sense of partial order
preservation by a state-space system (see, for example, \cite{Angeli2003a}).

Monotonicity is preserved under a number of operations. 
The proof of the following lemma may be found in \cite{Ryu2021a}.

\begin{lemma}
        \label{lem:monotone_properties}
        Consider operators $G$ and $F$ which are monotone on $\mathcal{H}$.  Then
        \begin{enumerate}
                \item $G^{-1}$ is monotone; \label{inversion}
                \item $G + F$ is monotone; \label{sum}
                \item $\alpha G$ is monotone for $\alpha > 0$.
        \end{enumerate}
\end{lemma}

Maximality is preserved under inversion.  However, in general, maximality is not 
preserved when two relations are added (indeed, 
their sum may be empty). We make the following assumption on summations throughout
the rest of this paper, which guarantees maximality of the sum, by \cite[Thm. 1]{Rockafellar1970}.

\begin{assumption}
        \label{ass:domains}
        Any summation of two operators $G$ and $F$ obeys
                        \begin{IEEEeqnarray*}{rrCl}
                               & \interior \dom F \cap \dom G &\neq& \varnothing\\
                                \text{or } & \interior \dom G \cap \dom F &\neq&
                                \varnothing,
                        \end{IEEEeqnarray*}
                        where $\dom S$ denotes the domain of the operator $S$, and
                        $\interior$ denotes the interior of a set.
\end{assumption}

This assumption is sufficient (but not necessary) for the existence of solutions to
the summation (that is, the resulting relation is nonempty).  We omit the proof of this fact.

\subsection{Stronger monotonicity properties}
\label{sec:more_properties}

\begin{definition}
        An operator $S: \mathcal{H}\rightrightarrows\mathcal{H}$ has a \emph{Lipschitz constant of} $\lambda>0$, or is
        \emph{$\lambda$-Lipschitz} if, for all
        $(u_1, y_1), (u_2, y_2) \in \operatorname{gra} S$,
        \begin{equation*}
                \norm{y_1 - y_2} \leq \lambda\norm{u_1 - u_2}.
        \end{equation*}
        If $\lambda < 1$, $S$ is called a \emph{contraction}.  If $\lambda = 1$, $S$ is called
        \emph{nonexpansive}.
\end{definition}
Note that if $S$ is $\lambda$-Lipschitz, it is also $\bar\lambda$-Lipschitz for all
$\bar\lambda > \lambda$.

\begin{definition}
        Given $\mu > 0$, an operator $S: \mathcal{H}\rightrightarrows\mathcal{H}$ is $\mu$-\emph{coercive} or
        $\mu$-\emph{strongly monotone} if, for all $(u_1, y_1), (u_2, y_2) \in
        \operatorname{gra} S$,
        \begin{equation*}
                \ip{u_1 -u_2}{y_1 - y_2} \geq \mu\norm{u_1 - u_2}^2.
        \end{equation*}
        $S$ is called $\mu$-\emph{hypomonotone} in the case that $\mu < 0$.
        If the sign of $\mu$ is unknown, we simply say $S$ is
        $\mu$-\emph{monotone}.
\end{definition}

\begin{definition}
        Given $\gamma > 0$, an operator $S: \mathcal{H}\rightrightarrows\mathcal{H}$ is $\gamma$-\emph{cocoercive} 
        if, for all $(u_1, y_1), (u_2, y_2) \in \operatorname{gra} S$,
        \begin{equation*}
                \ip{u_1 - u_2}{y_1 - y_2} \geq \gamma \norm{y_1 - y_2}^2.
        \end{equation*}
        $S$ is called $\gamma$-\emph{cohypomonotone} in the case that $\gamma < 0$.
\end{definition}

It is seen immediately that $F$ is $\mu$-coercive if and only if $F^{-1}$ is
$\mu$-cocoercive.  It also follows from the Cauchy-Schwarz inequality that $F$
has a Lipschitz constant of  $1/\gamma$ if $F$ is $\gamma$-cocoercive.  Finally, if
$A$ is $\mu$-coercive (resp. $\gamma$-cocoercive) and $B$ is monotone, $A + B$ is
is $\mu$-coercive (resp. $\gamma$-cocoercive).  For more details on these
properties, we refer the reader to \autocite[\S 2.2]{Ryu2021a} and
\autocite{Giselsson2021}.

\section{Monotone one-port circuits}
\label{sec:1ports}

The systems considered in this paper are electrical one-port circuits.  Analogous to
the classical realization of passive LTI transfer functions by series/parallel
interconnections of resistors, capacitors and inductors \autocite{Bott1949}, we
consider the class of systems which can be realized as series/parallel
interconnections of basic one-port elements, which are modelled as monotone operators.  

One-port circuits have two external terminals.  The port voltage $v$ may be measured
across these terminals, and the port current $i$ may be measured through them. We assume that each of
these variables takes values in $\R$.  
A one-port circuit $E$ is defined
by a relation on $L_{2, T}$ between current and voltage.  We denote by $d(E) \in \{i\to v, v \to i\}$ the
direction of the relation $E$, either current to voltage (current controlled) or
voltage to current (voltage controlled).  We will often denote current controlled
circuits by $R$, and voltage controlled circuits by $G$.  These may in general be
arbitrary impedances or admittances, and are not restricted to being memoryless.
We say that $E$ is a \emph{$\mu$-monotone one-port} if it is defined by a $\mu$-monotone relation.

\subsection{Series and parallel interconnections}

Two one-ports may be combined to build a new one-port by series or parallel
interconnection.  These are illustrated in Figure~\ref{fig:series-parallel-circuit}.
\begin{figure}[ht]
         \centering
         \includegraphics{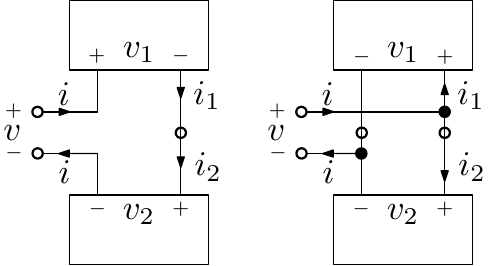}
         \caption{Series (left) and parallel (right) interconnections of two 1-ports.}
         \label{fig:series-parallel-circuit}
\end{figure}

When two one-ports are
connected in parallel, their relations must be from voltage to current.  If they are
not, one or both relations must be inverted before interconnection.  Let $G_1$ and
$G_2$ be two one-port circuits such that $d(G_1) = d(G_2) = v \to i$.
For a parallel interconnection, the composition of Kirchoff's laws and the relations
$G_1$ and $G_2$ creates a natural forward relation from voltage to current, as
follows.
\begin{enumerate}
        \item KVL: $v = v_1 = v_2$
        \item Device: $(v_1, i_1) \in G_1$, $\quad(v_2, i_2) \in G_2$
        \item KCL: $i_1 + i_2 = i$.
\end{enumerate}
We therefore have a new relation $G = G_1 + G_2$, $d(G) = v \to i$.
This is illustrated in the left of Figure~\ref{fig:parallel-block}.  Calculating the inverse
relation, we have
\begin{IEEEeqnarray*}{rCl}
        i &\in& (G_1 + G_2)(v)\\
        G_1(v) &\in& i - G_2(v)\\
        v &\in& G_1^{-1}(i-G_2(v)),
\end{IEEEeqnarray*}
which is the negative feedback interconnection of $G_1^{-1}$ and $G_2$, illustrated in the right of Figure~\ref{fig:parallel-block}.

For a series interconnection, the roles of current and voltage are reversed.  Letting
$R_1$ and $R_2$ be two one-port circuits such that $d(R_1) = d(R_2) = i \to v$, their
series interconnection gives a relation from current to voltage, as follows. 
\begin{enumerate}
        \item KCL: $i = i_1 = i_2$
        \item Device: $(i_1, v_1) \in R_1$, $\quad(i_2, v_2) \in R_2$
        \item KVL: $v_1 + v_2 = v$.
\end{enumerate}
The new relation is $R = R_1 + R_2$, with $d(R) = i \to v$.  The inverse relation,
from $v$ to $i$, is the negative feedback interconnection of $R_1^{-1}$ and $R_2$.
This is illustrated in Figure~\ref{fig:series-block}.  Properties of a parallel
interconnection always hold for a series interconnection when the roles of $i$ and
$v$ are exchanged; as such, in several results which follow, we will state results for
parallel interconnections only.

\begin{figure}
        \centering
        \includegraphics{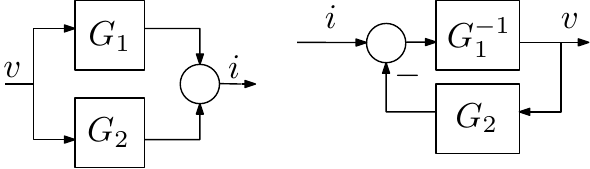}
        \caption{Block diagram of parallel interconnection, illustrating parallel
        forward relation from voltage to current, and negative feedback relation from
        current to voltage.}
        \label{fig:parallel-block}
\end{figure}
\begin{figure}
        \centering
        \includegraphics{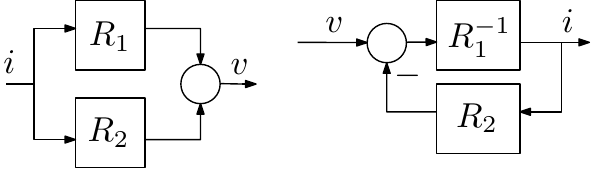}
        \caption{Block diagram of series interconnection, illustrating
        parallel forward relation from current to voltage, and negative feedback
        relation from voltage to current.}
        \label{fig:series-block}
\end{figure}

A standard result is that monotonicity of circuits is preserved under series and
parallel interconnection \autocite{Passty1986}.
Precisely, we have the following.
\begin{proposition}\label{prop:series_parallel}
         \begin{enumerate}
                 \item Let $E_1$ and $E_2$ be monotone one-port circuits such that
                         $d(E_1), d(E_2) \in \{i\to v, v \to i\}$. Then the series and
                         parallel interconnections of $E_1$ and $E_2$ are both
                         monotone one-ports. \label{prop:series_parallel:one}
                 \item Let $G_1$ and $G_2$ be one-port circuits such that $G_1$ is
                         $\alpha$-monotone, $G_2$ is $\beta$-monotone, and
                         $d(G_1) = d(G_2) = v \to i$. Then the parallel
                         interconnection of $G_1$ and $G_2$ is $(\alpha +
                         \beta)$-monotone. \label{prop:series_parallel:two}
         \end{enumerate}
\end{proposition}
\begin{proof}
        The proof of Part~\ref{prop:series_parallel:one} follows directly from the
        preservation of monotonicity under inversion and addition
        (Lemma~\ref{lem:monotone_properties}).
        Part~\ref{prop:series_parallel:two} 
        follows from the fact that if $E_1$ is $\alpha$-monotone and $E_2$ is
        $\beta$-monotone, $E_1 + E_2$ is $(\alpha + \beta)$-monotone:
        \begin{IEEEeqnarray*}{+cl+x*}
                &\ip{u_1 - u_2}{(E_1 + E_2)u_1 - (E_1 + E_2)u_2} \\
                =& \ip{u_1 - u_2}{E_1u_1 - E_1u_2} + \ip{u_1 - u_2}{E_2u_1 - E_2u_2}\\
                \geq& (\alpha + \beta)\norm{u_1 - u_2}^2. &\qedhere
        \end{IEEEeqnarray*}
\end{proof}

Repeatedly applying series and parallel interconnections allows a collection of
one-port circuits to be assembled into a single, larger one-port circuit,  using the
relational operations of inversion and addition.

We conclude this section by remarking that the preservation of monotonicity under
port interconnection, proved in Proposition~\ref{prop:series_parallel}, can be
reinterpreted in terms of negative feedback.  As shown in
Figures~\ref{fig:parallel-block}, the negative feedback interconnection
of  two operators $F$ and $G$ can be represented as a parallel interconnection of $F^{-1}$ and $G$.
Proposition~\ref{prop:series_parallel} then allows us to recover the incremental form
of the fundamental theorem of passivity.

\begin{corollary}
        Given two operators $F$ and $G$, each monotone on a Hilbert space
        $\mathcal{H}$, their negative feedback interconnection $(F^{-1} +
        G)^{-1}$ is monotone on $\mathcal{H}$.
\end{corollary}

\section{Algorithmic steady-state analysis of series/parallel monotone one-ports}
\label{sec:computation}
In this section, we develop an algorithmic method for computing the periodic
input/output behavior 
of a monotone one-port. We consider a circuit
made of series and parallel interconnections of one-port elements, each defining a
(discrete time) monotone operator on $l_{2, T}$.  The circuit defines a monotone
operator $M$.  Concrete examples of such circuits are given in Sections~\ref{sec:RLC}
and~\ref{sec:conductance}.

Without loss of generality, we consider the problem of computing the ``output''
current $i^\star$ of the monotone operator $M$ corresponding to an ``input'' voltage
$v^\star$.

We compute the solution as the fixed point of an iterative splitting algorithm
determined from the series and parallel structure of the circuit.  The algorithm is
first presented for two elements, then generalized to an arbitrary composition of
series and parallel interconnections.

\subsection{Splitting algorithms for two element circuits} \label{sec:splitting}
\label{sec:fixed_point_algorithms}

There is a large body of literature on splitting algorithms, which solve problems of the form $0 \in M_1(u) +
M_2(u)$, where $M_1 + M_2$ is a maximal monotone relation.  If $M$ consists of two
elements, connected in series or parallel, we can convert our problem to this form by
writing $0 \in M_1(i) + M_2(i) - v^\star$ (assuming a series interconnection - the
parallel interconnection is obtained by exchanging $i$ and $v$).  The offset
$-v^\star$ does not affect the monotonicity properties of $M$.  Splitting algorithms
distribute the computation on the
components $M_1$ and $M_2$.  They are useful when computation for the individual components is
easy, but computation for their sum is hard.  Here, we describe two splitting algorithms - the
forwards/backwards splitting, and the Douglas-Rachford splitting.
Given an operator $S$ and a
scaling factor $\alpha$, the $\alpha$-resolvent of $S$ is defined to be the operator
\begin{IEEEeqnarray*}{rCl}
        \res_{\alpha S} \coloneqq (I + \alpha S)^{-1}.
\end{IEEEeqnarray*}
If $S$ is maximal monotone, $\res_{S}$ is single-valued \autocite{Minty1961}.

\subsubsection*{Forward/backward splitting}

This is the simplest splitting algorithm \autocite{Passty1979, Gabay1983, Tseng1988}.  Suppose $M_1$ and $\res_{\alpha M_2}$
are single-valued.  Then:
\begin{IEEEeqnarray*}{lrCl}
       & 0 &\in& M_1(x) + M_2(x)\\
        \iff & 0 & \in & x - \alpha M_1(x) - (x + \alpha M_2(x))\\
        \iff &(I + \alpha M_2)x &\ni& (I - \alpha M_1)x \\
        \iff & x &=& \res_{\alpha M_2} (I - \alpha M_1) x.
\end{IEEEeqnarray*}
The fixed point iteration $x^{j+1} = \res_{\alpha M_2}(x^j - \alpha M_1(x^j))$ is the forward/backward splitting algorithm. 
When $M_1$ and $M_2$ are monotone, the convergence conditions for the
forward/backward splitting are standard in the literature (see, for instance,
\autocite{Ryu2016}).  
\textcite[$\S6$]{Giselsson2021} generalize these conditions to cases where $M_1 +
M_2$ is monotone, but either $M_1$ or $M_2$ is hypomonotone.  These conditions may be summarised as follows.

\begin{proposition}
        \label{prop:forward-backward}
        Let $\mu \geq 0$, $\omega \geq 0$ and $\beta > 0$, and $M_1$ and $M_2$ be
        operators on a Hilbert space $\mathcal{H}$.
        The forward/backward algorithm, with scaling factor $\alpha \in (0, 2/(\beta + 2\mu))$, converges to a zero of $M_1 + M_2$, if one
        exists, in each of
        the following cases:
        \begin{itemize}
                \item $M_1$ is maximally $\mu$-monotone, $M_1 - \mu I$ is
                        $1/\beta$-cocoercive, $M_2$ is maximally $(-\omega)$-monotone
                        and $\mu \geq \omega$.
                \item $M_1$ is maximally $(-\omega)$-monotone, $M_1 + \omega I$ is
                        $1/\beta$-cocoercive, $M_2$ is maximally $\mu$-monotone
                        and $\mu \geq \omega$.
                \item $M_1$ is $\beta$-Lipschitz, $M_2$ is maximally $\mu$-monotone and
                        $\mu \geq \beta$.
        \end{itemize}
\end{proposition}

\subsubsection*{Douglas-Rachford splitting}

This most successful splitting algorithm 
forms the basis of the Alternating Direction Method of
Multipliers \cite{Boyd2010, Douglas1956, Lions1979}.

The reflected resolvent, or Cayley operator, is the operator
\begin{IEEEeqnarray*}{rCl}
        R_{\alpha S} \coloneqq 2\res_{\alpha S} - I.
\end{IEEEeqnarray*}
Given two operators $M_1$ and $M_2$, a scaling factor $\alpha$ and an initial value
$z^0$, 
the Douglas-Rachford algorithm is the iteration in $k$ given by
\begin{IEEEeqnarray*}{rCl}
z^{k + 1} &=& T(z^{k}),\\
x^k &=& \res_{\alpha M_2}{z^k},
\end{IEEEeqnarray*}
where $T$ is given by
\begin{IEEEeqnarray}{rCl}
T = \frac{1}{2}(I + R_{\alpha M_1} R_{\alpha M_2}).\label{eq:DR_operator}
\end{IEEEeqnarray}

\textcite[Thm 5.1]{Giselsson2021} give the most general conditions for convergence of
the Douglas-Rachford algorithm, again allowing one operator in the sum to be
hypomonotone.

\begin{proposition} \label{thm:DR_convergence}
        Let $M_1$ and $M_2$ be operators on a Hilbert space $\mathcal{H}$.  Let $\mu
        > \omega \geq 0$ and $\alpha \in (0, (\mu - \omega)/2\mu\omega)$.  The Douglas-Rachford
        algorithm converges to a zero of $M_1 + M_2$, if one exits, in each of the
        following cases.
        \begin{itemize}
        \item $M_1$ is maximally $(-\omega)$-monotone and $M_2$ is maximally
                $\mu$-monotone.
        \item $M_2$ is maximally $(-\omega)$-monotone and $M_1$ is maximally
                $\mu$-monotone.
\end{itemize}
\end{proposition}

\subsection{A nested splitting algorithm for three element 
circuits}\label{sec:three}

If $M$ is composed of three elements, with one series interconnection and one
parallel interconnection (see Figure~\ref{fig:three-circuit}), $M$ has the form
$M = M_3 + (M_2 + M_1)^{-1}$, and we can convert our problem to the form $0 \in (M_3 + (M_2 +
M_1)^{-1})(u)$ again by offsetting by the input current or voltage. One approach to
solving this problem is to using a splitting algorithm such as the forward/backward algorithm,
with the resolvent step applied for $M_3$ and the forward step applied for $(M_2 +
M_1)^{-1}$.  Applying this forward step amounts to solving $y = (M_2 + M_1)^{-1}(u)$
for some $u$, which may be rewritten as $0 \in (M_2 + M_1)(y) - u$.  This can be
solved by again applying the forward/backward algorithm.

\begin{figure}[hb]
        \centering
        \includegraphics{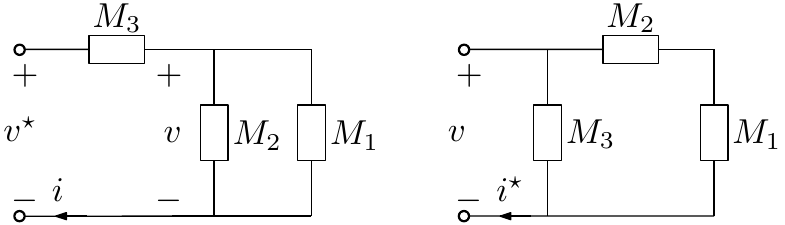}
        \caption{The two possible configurations of three elements with one series
        interconnection and one parallel interconnection.}%
        \label{fig:three-circuit}
\end{figure}

This procedure has poor complexity: for every forward/backward
step for $M_3 + (M_2 + M_1)^{-1}$, an \emph{entire} fixed point iteration has to
be computed for (an offset version of) $M_2 + M_1$.  In this section, we propose an
alternative procedure.  Rather than applying a forward step for the relation $(M_2 +
M_1)^{-1}$, we simply apply a \emph{single step} of the fixed point iteration needed
to compute this forward step, using the forward/backward algorithm.  Assume, without
loss of generality, that $d(M_3) = i \to v$ and $d(M_2) = d(M_1) = v \to i$ (the
configuration shown on the left of Figure~\ref{fig:three-circuit}).  Suppose that
$v^\star \in (M_3 + (M_2 + M_1)^{-1})(i)$.  Assume that $M_1$, $\res_{\alpha_1 M_2}$ 
and $\res_{\alpha_2 M_3}$ are single-valued.  We then have:
\begin{IEEEeqnarray}{rCl}
        v^\star &\in& v + M_3(i)\label{eq:i_step}\\
        v &\in& (M_2 + M_1)^{-1}(i),\label{eq:v_step} 
\end{IEEEeqnarray}
where $v$ is the voltage over $M_2$, illustrated on the left of
Figure~\ref{fig:three-circuit}.  Equation~\eqref{eq:i_step} gives
\begin{IEEEeqnarray*}{rCl}
        i + \alpha_2 M_3(i) &\ni& i - \alpha_2 v + \alpha_2 v^\star\\
        i &=& (I + \alpha_2 M_3)^{-1}(i - \alpha_2 v + \alpha_2 v^\star)\\
        i &=& \res_{\alpha_2 M_3}(i - \alpha_2 v + \alpha_2 v^\star).
\end{IEEEeqnarray*}
Equation~\eqref{eq:v_step} gives
\begin{IEEEeqnarray*}{rCl}
        i &\in& (M_2 + M_1)(v)\\
        v + \alpha_1 M_2(v) &\ni& v - \alpha_1 M_1 (v) + \alpha_1 i\\
        v &=& (I + \alpha_1 M_2)^{-1}(v - \alpha_1 M_1(v) + \alpha_1 i)\\
        v &=& \res_{\alpha_1 M_2}(v - \alpha_1 M_1(v) + \alpha_1 i).
\end{IEEEeqnarray*}
This shows that a fixed point of the iteration
\begin{IEEEeqnarray}{rCl}\IEEEyesnumber\label{eq:three_iteration}\IEEEyessubnumber
        v^{k + 1} &=& \res_{\alpha_1 M_2} (v^k - \alpha_1 M_1(v^k) + \alpha_1 i^k)\\
        i^{k + 1} &=& \res_{\alpha_2 M_3} (i^k - \alpha_2 v^{k+1} + \alpha_2
        v^\star)\IEEEyessubnumber\label{eq:three_second}
\end{IEEEeqnarray}
is a solution to our original problem.
In the next section, we generalize this algorithm to an arbitrary series/parallel
monotone one-port, and in Theorem~\ref{thm:nested_convergence}, we give a general
condition under which the algorithm is guaranteed to converge to such a fixed point.
The three element circuits of this section are revisited in
Example~\ref{ex:three_convergence}.

\subsection{A nested splitting algorithm for arbitrary series/parallel
circuits}\label{sec:nested}

In this section, we introduce a new splitting algorithm, the \emph{nested
forward/backward algorithm}, which generalizes the algorithm described in the
previous section to monotone one-ports with
arbitrary series and parallel interconnections, which have the general form shown in
Figure~\ref{fig:nested_algo} (allowing elements to be open circuits, short circuits,
or whole subcircuits).  We assume for simplicity that the relations $G_j$ and $R_j$
are single-valued, although the extension to multi-valued relations
is straightforward.

\begin{figure}[hb]
        \centering
        \includegraphics[width=\linewidth]{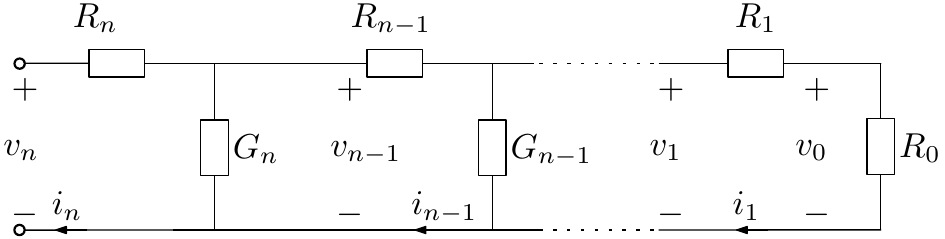}
        \caption{Circuit structure with nested series and parallel interconnections.
        $R_n$ represents a one-port whose $i-v$ relation is known, $G_n$ represents a
one-port whose $v-i$ relation is known.}%
\label{fig:nested_algo}
\end{figure}

The $v-i$ relation of the circuit in Figure~\ref{fig:nested_algo} is given by
\begin{IEEEeqnarray}{rCl}
i_n &=& (R_n {+} (G_n {+} (\ldots {+} (R_1 {+} R_0)^{-1}\ldots
)^{-1})^{-1})^{-1}(v_n).\label{eq:nested_relation}
\end{IEEEeqnarray}
If each inversion is solved using a fixed point iteration, the number of fixed points 
that must be computed scales with order $\mathcal{O}(m^n)$, where $n$ is the number
of inverses in Equation~\eqref{eq:nested_relation}, and $m$ is the number of steps
needed to compute each inverse.  Following the argument of the
previous section, the nested forward/backward
algorithm, given in Algorithm~\ref{alg:nested}, solves equations of the form \eqref{eq:nested_relation} by
replacing inverse operators with a single step of the forward/backward iteration
needed to compute them.  In this way, every inversion is computed simultaneously,
using a single fixed point algorithm.  

\algdef{SE}[DOWHILE]{Do}{DoWhile}{\algorithmicdo}[1]{\algorithmicwhile\ #1}%

\begin{algorithm}
        \caption{Nested Forward/Backward Algorithm}\label{alg:nested}
        \begin{algorithmic}[1]
                \State \textbf{Data:} Series/parallel one-port with $2n$ elements,
                step sizes $\alpha_j > 0$, $j = 1, \ldots, 2n-1$, external signal
                $v_n$, convergence tolerance $\epsilon > 0$.
                \For{$j = 1, \ldots, n$}
                        \State Initialize $v^1_{j-1}$, $i^1_j$.
                \EndFor
                \State $k = 1$
                \Do
                        \State $i_1^{k+1} = \res_{\alpha_1 R_1}(i_1^k - \alpha_1 R_0(i_1^k) +
                        \alpha_1 v_1^k)$
                        \label{alg:i_1_update}
                        \For{$j = 2, \ldots, n$}
                        \State $v_{j-1}^{k + 1} = \res_{\alpha_{2j-2} G_{j}}
                                (v_{j-1}^k - \alpha_{2j-2}i_{j-1}^{k+1} + \alpha_{2j-2}i_j^k)$
                                \label{alg:v_update}
                        \State $i_{j}^{k + 1} = \res_{\alpha_{2j-1} R_{j}}
                                (i_{j}^k - \alpha_{2j-1}v_{j-1}^{k+1} + \alpha_{2j-1}v_j^k)$
                                \label{alg:i_update}
                        \EndFor
                        \State $k = k+1$.
                        \DoWhile{$\max_j (|v_j^{k + 1} - v^k_j|, |i_j^{k + 1} -
                        i^k_j|) > \epsilon$}
        \end{algorithmic}
\end{algorithm}

\begin{theorem}\label{thm:nested_convergence}
        Algorithm~\ref{alg:nested} converges to a solution of
        Equation~\ref{eq:nested_relation} as $k \to \infty$ if $R_0$ is coercive and Lipschitz, all the
         $R_j, G_j$ are monotone for $j = 1, \ldots, n$, and the eigenvalues of
         $\mathcal{A}$ all lie within the unit circle, where $\mathcal{A}$ is defined
         as the $(2n-1)\times(2n-1)$ matrix with columns
         \begin{IEEEeqnarray*}{l}\label{eq:nested_matrix}
        \begin{pmatrix}
\beta_1 \gamma_1 \\                                                   
\alpha_2 \beta_1 \gamma_2 \gamma_1  \\                               
\alpha_3 \alpha_2 \beta_1\gamma_3 \gamma_2 \gamma_1 \\                 
\alpha_4 \alpha_3 \alpha_2 \beta_1 \gamma_4\gamma_3\gamma_2\gamma_1\\ \vdots
        \end{pmatrix},\;
        \begin{pmatrix}
\gamma_1 \alpha_1\\                                                       
\gamma_2(1  +\gamma_1 \alpha_1 \alpha_2) \\                               
\alpha_3\gamma_2\gamma_3(1 +\gamma_1\alpha_1 \alpha_2) \\                
\alpha_4 \alpha_3 \gamma_4\gamma_3\gamma_2(1 + \gamma_1\alpha_1\alpha_2)\\ \vdots
\end{pmatrix},\\
\begin{pmatrix}
0 \\                                                       
\gamma_2 \alpha_2 \\                                       
\gamma_3(1 + \gamma_2 \alpha_2 \alpha_3) \\                
\alpha_4\gamma_4 \gamma_3(1 + \gamma_2\alpha_2 \alpha_3)\\ \vdots
\end{pmatrix},\;
                \begin{pmatrix}
 0\\ 
 0\\ 
 \gamma_3 \alpha_3 \\
 \gamma_4(1 + \gamma_3\alpha_3 \alpha_4)\\ \vdots
        \end{pmatrix},
        \end{IEEEeqnarray*}
         and so on, where, for $j=1, \ldots, n$, $\gamma_{2j-2}$ is a Lipschitz constant of
        $\res_{\alpha_{2j-2} R_j}$,  $\gamma_{2j-1}$ is a Lipschitz constant of
        $\res_{\alpha_{2j-1} G_j}$ and $\beta_1$ is a Lipschitz constant of the
        operator $(I - \alpha_1 R_0)$.
\end{theorem}

To clarify, the constants $\alpha_j$ may be chosen to tune the convergence rate of
the algorithm.  The constants $\beta_1$ and $\gamma_j$ must be Lipschitz constants
for the relevant operators. Coercivity and Lipschitz
continuity of $R_0$, say with constants $\mu$ and $\lambda$ respectively, means that
$\alpha_1$ can always be chosen so that $0 < \beta_1 < 1$ by the formula $\beta_1 = 1
- 2\mu\alpha_1 + \lambda^2 \alpha^2$ \autocite[p. 39]{Ryu2021a}. 
Monotonicity of $R_j$ and $G_j$ for all $j$ implies that
all resolvents used in the algorithm are nonexpansive, so the $\gamma_j$ are at most
equal to $1$.  If, furthermore, an element $R_j$ or $G_j$ is  $\mu_j$-coercive, we have
$\gamma_j = 1/(1 + \alpha_j\mu_j)$.  In the limiting case $\alpha_j = 0$ for all $j$,
the matrix $\mathcal{A}$ is the identity, and has all its eigenvalues on the boundary
of the unit disc.  Before giving the proof of
Theorem~\ref{thm:nested_convergence}, we revisit the three element circuits of
Section~\ref{sec:three}.

\begin{example}\label{ex:three_convergence}
        Consider the circuit shown in the left of Figure~\ref{fig:three-circuit}, consisting
        of three elements, $M_1$, $M_2$ and $M_3$.    Section~\ref{sec:three} describes
        a special case of Algorithm~\ref{alg:nested} for this circuit.  In this
        example, we apply Theorem~\ref{thm:nested_convergence} to this circuit, and
        give a convergence condition for the fixed point iteration given by
        Equation~\eqref{eq:three_iteration}.

        For this circuit, the matrix $\mathcal{A}$ is given by
        \begin{IEEEeqnarray*}{rCl}
                \begin{pmatrix}
                        \gamma_1\beta_1 & \gamma_1 \alpha_1 \\
                        \alpha_2\beta_1\gamma_1\gamma_2 & \gamma_2 +
                        \alpha_1\alpha_2\gamma_1\gamma_2
                \end{pmatrix},
        \end{IEEEeqnarray*}
        where $\beta_1$ is a Lipschitz constant for the operator $(I - \alpha_1
        M_1)$, $\gamma_1$ is a Lipschitz constant for the operator
        $\res_{\alpha_1M_2}$ and $\gamma_2$ is a Lipschitz constant for the operator
        $\res_{\alpha_2 M_3}$.

        Suppose $M_1$ is
        $\mu_1$-coercive and $\lambda_1$-Lipschitz, and $M_2$ and $M_3$ are $\mu_2$-
        and $\mu_3$-coercive, respectively.  For simplicity, set $\alpha_1 = \alpha_2
        = \alpha$. Then we can write
        \begin{IEEEeqnarray*}{rCl}
                \beta_1 &=& 1 - 2\mu_1\alpha + \lambda_1^2 \alpha^2\\
                \gamma_1 &=& 1/(1 + \alpha \mu_2)\\
                \gamma_2 &=& 1/(1 + \alpha \mu_3).
        \end{IEEEeqnarray*}
        Figure~\ref{fig:three_convergence} shows the magnitude of the largest
        eigenvalue of $\mathcal{A}$ as a function of $\alpha$, for $\mu_1 = \mu_2
        = \mu_3 = 1$ and $\lambda_1 = 2$.  It can be seen that, for $\alpha \in (0,
        0.5)$,
        the eigenvalues of $\mathcal{A}$ lie within the unit disc, and the fixed
        point iteration is a contraction mapping and is guaranteed to converge.  The
        best guaranteed convergence rate is obtained with $\alpha \approx 0.27$.
\end{example}

        \begin{figure}[hb]
                \centering
                \begin{tikzpicture}
                        \begin{axis}
                                [
                                no markers,
                                name = ax1,
                                width=0.35\textwidth,
                                height=0.3\textwidth,
                                ticklabel style={/pgf/number format/fixed},
                                ylabel={$|\lambda_{\max}( \mathcal{A})|$},
                                xlabel={$\alpha$},
                                cycle list name=colors,
                                grid=both,
                                grid style={line width=.1pt, draw=Gray!20},
                                axis x line=bottom,
                                axis y line=left
                                ]
                                \addplot table [x=a, y=l, col sep = comma, mark =
                                        none]{"./three_convergence.csv"};
                       \end{axis}
                \end{tikzpicture}
                \caption{The magnitude of the largest eigenvalue of $\mathcal{A}$ as
                a function of step size $\alpha$, for the fixed point iteration of
                Example~\ref{ex:three_convergence}. The curve lies below $1$ for $\alpha \in (0,
        0.5)$, with a minimum at $\alpha \approx 0.27$.  The curve is not symmetric.}%
                \label{fig:three_convergence}
        \end{figure}
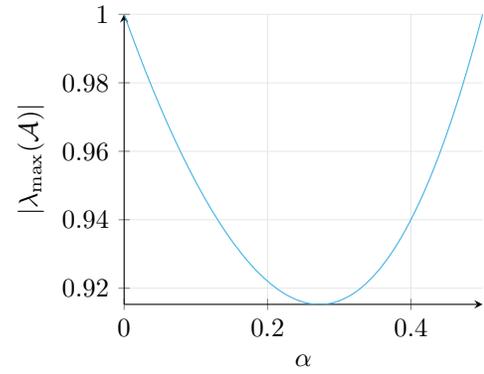

\begin{proof}[Proof of Theorem~\ref{thm:nested_convergence}]
        We begin by showing that a fixed point of the iteration in
        Algorithm~\ref{alg:nested} is a solution to
        Equation~\ref{eq:nested_relation}.
        Indeed, substituting $v_j^{k+1} = v_j^k$ and $i_j^{k+1} = i_j^k$ into
        lines~\ref{alg:i_1_update}, \ref{alg:v_update} and~\ref{alg:i_update} of
        Algorithm~\ref{alg:nested} gives
        \begin{IEEEeqnarray*}{rCl}
        v_1 &=& R_1(i_1) + R_0(i_1)\\
        i_j &=& G_j(v_{j-1}) + i_{j-1}\\
        v_j &=& R_j(i_j) + v_{j-1},
        \end{IEEEeqnarray*}
        from which we obtain
        \begin{IEEEeqnarray*}{rCl}
        i_1 &=& (R_1 + R_0)^{-1}(v_1)\\
        i_2 &=& (R_2 + (G_2 + (R_1 + R_0)^{-1})^{-1})^{-1} (v_2),
        \end{IEEEeqnarray*}
        and so on, to arrive at Equation~\ref{eq:nested_relation}, as required.

        We now show that Algorithm~\ref{alg:nested} converges to a fixed point under
        the stated conditions.  We simplify notation by defining $u_j = i_j$, $j$
        odd, and $u_j = v_j$, $j$ even.

        Let $u^{k}$ and $w^{k}$ be two sequences
        of iterates generated by Algorithm~\ref{alg:nested}, with the same input
        $u_n^k = w_n^k = v^\star$, and denote $u_j^k -
        w_j^k$ by $\Delta
        u_j^{k}$. It then follows from lines~\ref{alg:i_1_update}, \ref{alg:v_update} and~\ref{alg:i_update} of
        Algorithm~\ref{alg:nested} that, for $j = 1, \ldots, 2n-1$,
        \begin{IEEEeqnarray*}{rCl}
                \norm{\Delta u_1^{k+1}} \leq \gamma_1 \norm{\Delta u_1^k -\alpha_1
                        \Delta R_0(u_1^k) + \alpha_1 \Delta u_2^k}\\
                \norm{\Delta u_j^{k+1}} \leq \gamma_j \norm{\Delta u_j^k - \alpha_j
                \Delta u_{j-1}^{k+1} + \alpha_j \Delta u_{j+1}^k},
        \end{IEEEeqnarray*}
        from which it follows, via the triangle inequality, that
        \begin{IEEEeqnarray*}{rCl}
                \norm{\Delta u_1^{k+1}} \leq \gamma_1 \beta_1 \norm{\Delta u_1^k} +
                \gamma_1 \alpha_1 \norm{\Delta u_2^k}\\
                \norm{\Delta u_j^{k+1}} \leq \gamma_j \norm{\Delta u_j^k} + \gamma_j \alpha_j
                \norm{\Delta u_{j-1}^{k+1}} + \gamma_j \alpha_j \norm{\Delta u_{j+1}^k},
        \end{IEEEeqnarray*}
        where $\Delta u_n^k = 0$ for all $k$.  Let $n(\Delta u^k)$ denote the vector
        $(\norm{\Delta u_1^{k}}, \norm{\Delta u_2^{k}}, \norm{\Delta u_3^{k}},
        \norm{\Delta u_4^{k}}, \ldots)\tran$. It follows that
        \begin{IEEEeqnarray*}{rCl}
        n(\Delta u^{k+1}) &\leq& 
        \mathcal{A} n(\Delta u^k)\\
                          &\leq& \mathcal{A}^{k} n(\Delta u^1),
        \end{IEEEeqnarray*}
        where $\mathcal{A}$ is the matrix given in the statement of the theorem.  It
        follows from the nonnegativity of $n(\Delta u^{k+1})$ (or from the
        elementwise nonnegativity of $\mathcal{A}$) that $\mathcal{A}^{k}
        n(\Delta u^1)$ is elementwise nonnegative for all $k$.
        We then have $0 \leq n(\Delta u^{k+1}) \leq z^{k+1}$, where $z^{k+1}$ is the
        solution to the difference equation $z^{k+1} = \mathcal{A}z^k$ with initial
        condition $n(\Delta u^1)$.  Since the eigenvalues of $\mathcal{A}$ are within
        the unit circle, it is a standard result of linear systems theory
        that there exist a norm $\norm{\cdot}_P$ and rate $0 < \lambda < 1$  such that 
        $\norm{z^{k+1}}_P \leq \lambda \norm{z^k}_P$.  It follows that the sequence $n(\Delta u^k)$ converges to the zero vector 
        in the norm $\norm{\cdot}_P$ at least as fast as the sequence $z^k$.
        It then follows from the Banach fixed point theorem that each $u_j^k$ converges to a limit $u^\star_j$ as $k\to \infty$, which completes the proof.
\end{proof}

        Theorem~\ref{thm:nested_convergence} demonstrates the validity of
        Algorithm~\ref{alg:nested}, in that it proves the existence of circuits for
        which the algorithm is guaranteed to converge. In simple
        cases, like that of Example~\ref{ex:three_convergence}, the condition can
        furthermore be used as a
        design tool.  Theorem~\ref{thm:nested_convergence} raises several
        questions for future research.  The first is whether a general solution
        exists for the step sizes $\alpha_j$ which minimizes the spectral radius of
        $\mathcal{A}$.  Furthermore, Theorem~\ref{thm:nested_convergence} is
        conservative, relying on a small gain argument and proving contraction, a
        strong convergence property.  An interesting question is whether there exist
        less conservative convergence conditions.

\section{RLC circuits}
\label{sec:RLC}

Here, we consider one-port circuits formed by the series and parallel interconnection
of resistors, capacitors and inductors.  This is the class of circuits considered in 
the conference version of this paper \autocite{Chaffey2021}.

A resistor is a relation $R$ on $\R$, the \emph{device law}, between current and voltage:
\begin{IEEEeqnarray*}{rrCl}
        & R = \left\{(i, v) \in \R \times \R\; |\; v \in R(i)\right\}.
\end{IEEEeqnarray*}
A resistor defines a 1-port relation on $L_{2, T}$ by applying
the relation $R$ at each time:

\begin{IEEEeqnarray*}{rCl}
\mathcal{R} = \left\{(i, v) \in L_{2, T} \times L_{2, T} \; |\; (i(t), v(t)) \in R \text{ for all } t \right\}.
\end{IEEEeqnarray*}

Given $i, v \in L_{2, T}$, define the charge $q(t) = \int_0^t i(\tau) \dd \tau$ and the magnetic flux linkage
$\phi(t) = \int_0^t v(\tau) \dd \tau$.
A capacitor is a relation $\mathcal{C}$ on $L_{2, T}$ between  voltage and current, defined by a device law $C(\cdot):\R \to \R$ which maps voltage to charge.
We assume that $C$ is single-valued and differentiable.
\begin{IEEEeqnarray*}{rrCl}
        & \mathcal{C} = \left\{(v, i) \in L_{2, T} \times L_{2, T}\; |\;q = C(v),\;
        \td{}{t} q = i \right\}.\\
\end{IEEEeqnarray*}
An inductor is given by a relation $\mathcal{L}$ on $L_{2, T}$ between the current
and voltage, defined by a device law $L(\cdot):\R\to \R$ which maps current to
magnetic flux linkage.  Again we assume $L$ to be single-valued and differentiable.
\begin{IEEEeqnarray*}{rrCl}
        & \mathcal{L} = \left\{(i, v) \in L_{2, T} \times L_{2, T}\; |\; \phi =
        L(i),\; \td{}{t} \phi = v \right\}.\\
\end{IEEEeqnarray*}

The following proposition shows that resistors map $T$-periodic inputs
to $T$-periodic outputs, capacitors map $T$-periodic voltages to $T$-periodic
currents, and inductors map $T$-periodic currents to $T$-periodic voltages.

\begin{proposition}\label{thm:periodic_elements}
        Single-valued memoryless relations and the derivative map $T$-periodic inputs to
        $T$-periodic outputs.
\end{proposition}

\begin{proof}
        Let $f$ be a single-valued memoryless relation, that is, a relation
        between $u$ and $y$ such that $y(t) = f(u(t))$.  Then $y(t + T) = f(u(t
        + T)) = f(u(t)) = y(t)$. 

        The property also holds for the derivative: 
        \begin{IEEEeqnarray*}{+rCl+x*}
                \td{u(t)}{t} &=& \lim_{h \rightarrow 0} \frac{u(t) + u(t + h)}{h}\\
                             &=& \lim_{h \rightarrow 0} \frac{u(t + T) + u(t + T + h)}{h}\\
                             &=& \td{u(t + T)}{t}.&\qedhere
        \end{IEEEeqnarray*}
\end{proof}

The following proposition gives a characterization of the monotonicity of 
resistors on $L_{2, T}$ in terms of their devices laws.
\begin{proposition}
        \label{lem:monotone-r}
                A resistor is monotone on $L_{2, T}$ if and only if its device law defines a monotone relation
                        on $\R$ between $i(t)$ and $v(t)$ for all $t$.
\end{proposition}

\begin{proof}
        \emph{If:} By monotonicity of the device law on $\R$, we have
        \begin{IEEEeqnarray*}{rCl}
                (i_1(t) - i_2(t))(v_1(t) - v_2(t)) \geq 0 \text{ for all } t,
        \end{IEEEeqnarray*}
        from which it follows that
        \begin{IEEEeqnarray*}{rCl}
                \ip{i_1 - i_2}{v_1 - v_2} &=& \int_{0}^{T} (i_1(t) -
                i_2(t))(v_1(t) - v_2(t)) \dd{t}\\
                                               &\geq& 0.
        \end{IEEEeqnarray*}
        \emph{Only if:} Assume by contradiction that the device law is not monotone
        on $\R$, that is, there exist $\iota_1, \iota_2 \in \R$ such that 
        \begin{IEEEeqnarray*}{rCl}
                (\iota_1 - \iota_2)(R(\iota_1) - R(\iota_2)) < 0.
        \end{IEEEeqnarray*}
        Taking the constant signals $i_1(t) = \iota_1$, $i_2(t) = \iota_2$ on $L_{2,
        T}$ shows that the resistor is not monotone on $L_{2, T}$.
\end{proof}

A natural question is whether the same can be said for inductors and capacitors -- are
these devices monotone if their device laws $C$ and $L$ are monotone?  A striking result of
\textcite{Kulkarni2001} is that this is true if and only if the device laws are
\emph{linear}.

\begin{proposition}
        \label{lem:monotone-l-c}
        Capacitors and inductors with monotone device laws on $\R$ are monotone on
        $L_{2, T}$ for all $T \geq 0$ if and only if their device laws are linear.
\end{proposition}

\begin{proof}
        The result is given by \autocite[Lemma~A.2]{Kulkarni2001}, noting that the
        signals used in their proof (Equation~A.4) are truncated square waves, which
        are signals on $L_{2, T}$ for $T$ equal to the length of the truncation.
\end{proof}

We now collect some results which show that, under mild conditions, series/parallel RLC circuits
define operators on $L_{2, T}$.

\begin{proposition}\label{thm:periodic_series_parallel}
        A parallel interconnection of $n$ one-ports which map $T$-periodic
        voltages to $T$-periodic currents also maps $T$-periodic voltages
        to $T$-periodic currents.  
\end{proposition}

\begin{proof}
        Periodicity is preserved under summation of signals, and therefore
        preserved by Kirchoff's laws.  Indeed, if $y(t) = u_1(t) + u_2(t)$, and $u_1$
        and $u_2$ are both $T$-periodic, then $y(t + T) = u_1(t+T) + u_2(t+T) =
        u_1(t) + u_2(t) = y(t)$.
\end{proof}

Next, we show that one-port circuits which obey simple conditions on their 
interconnections map periodic inputs to periodic outputs.
Other classes of systems with this property include contractive state space systems
\autocite{Sontag2010} and approximately finite memory input/output maps
\autocite{Sandberg1992}.

\begin{theorem}\label{thm:periodic}
        Let $M$ be the relation on $L_{2, T}$, from either $v$ to $i$ or $i$ to $v$, of a
        1-port constructed from the series and parallel interconnection of $n$
        constituent one-ports $M_i$, such that the
        construction obeys the following conditions
        \begin{enumerate}
                \item $M_i:L_{2, T} \to L_{2, T}$ for all $i$;\label{condition:dom}
                \item any one-port which must be inverted during the construction is
                        coercive and Lipschitz.\label{condition:resistor}
        \end{enumerate}
        Then $M$ maps any input in $L_{2, T}$ to a unique output in
        $L_{2, T}$.
\end{theorem}

\begin{proof}
        By assumption, each of the
        relations $M_i$ maps $T$-periodic inputs to $T$-periodic outputs (we denote
        this property by PIPO for the remainder of this proof).  We show that
        constructing a circuit under the given conditions preserves this property.
        This amounts to showing that the PIPO property is preserved under summation
        and inversion.  We have already observed that it is preserved under summation
        in Proposition~\ref{thm:periodic_series_parallel}. It remains to show
        that inversion preserves the PIPO property if the one-port to be inverted is
        coercive and Lipschitz.  Let $F$ be $\mu$-coercive and $\lambda$-Lipschitz.
        We will proceed by showing that $F$ is invertible on $L_{2, T}$: for every
        $y^\star \in L_{2, T}$, there exists $u^\star \in L_{2, T}$ such that
        $u^\star \in F^{-1}(y^\star)$.  This shows that $F^{-1}$ maps periodic
        $y^\star$ to periodic $u^\star$, that is, has the PIPO property.

        Let $y^\star \in L_{2, T}$ be arbitrary and define the incremental relation $\Delta F(u) =
        F(u) - y^\star$ on $L_{2, T}$. $\Delta
        F$ has the same coercivity and Lipschitz properties as $F$, independent of
        the choice of $y^\star$.  
        Given $\gamma > 0$, it is straightforward to check that the operator $I - \gamma \Delta F$ is an operator on
        $L_{2, T}$.
        We show that this operator is  a contraction mapping on $L_{2, T}$
        for small enough $\gamma > 0$, using the standard argument for proving
        convergence of the forward step algorithm \autocite[$\S$2.4.3]{Ryu2022}.  Indeed,
        \begin{IEEEeqnarray*}{l}
        \norm{(I - \gamma \Delta F)(x) - (I - \gamma\Delta F)y}^2 \\= \norm{x - y}^2 - 
        2\gamma \ip{x - y}{\Delta Fx - \Delta Fy} + \gamma^2\norm{\Delta Fx -
        \Delta Fy}^2\\
          \leq \left(1 - 2\gamma \mu + \gamma^2\lambda^2 \right) \norm{x - y}^2,
        \end{IEEEeqnarray*}
        where the inequality follows from the definitions of coercive and
        Lipschitz operators.  Solving $0 < \left(1 - 2\gamma \mu +
        \gamma^2\lambda^2 \right) < 1$ gives an allowable range of $\gamma \in (0,
        2\mu/\lambda^2)$ for $I - \gamma \Delta
        F$ to be a contraction mapping on $L_{2, T}$. It then follows from the Banach fixed point
        theorem that $I - \gamma \Delta F$ has a unique fixed point $u^\star \in L_{2, T}$ \autocite[\S
        2.4.2]{Ryu2021a}, \autocite{Banach1922}:
\begin{IEEEeqnarray*}{lrCl}
&u^\star &=& u^\star - \gamma \Delta F(u^\star)\\
\iff &\Delta F(u^\star) &=& F(u^\star) - y^\star = 0\\
\iff & F(u^\star) &=& y^\star.
\end{IEEEeqnarray*}
        As $y^\star \in L_{2, T}$ is arbitrary, this shows that
        $F$ is invertible on $L_{2, T}$. 
\end{proof}
When applied to RLC circuits, condition~\ref{condition:dom} of Theorem~\ref{thm:periodic}
requires capacitors to be connected in parallel and inductors to be connected in series. 

Nonlinear RLC circuits have also been recently studied as port interconnections of
incrementally port-Hamiltonian systems \autocite{Camlibel2022}.  While
port-Hamiltonian systems are always passive, the authors observe that incrementally
port-Hamiltonian systems are not always monotone.  They arrive at the same conclusion
that those RLC circuits that are monotone are precisely those with nonlinear monotone
resistors and LTI capcitors and inductors.

An interesting question is whether the class of monotone resistors, capacitors and
inductors can be extended beyond those allowed by Propositions~\ref{lem:monotone-r}
and~\ref{lem:monotone-l-c}.  One possibility is to allow time-varying energy storage
devices.
\textcite{Georgiou2020} define time-varying, or adjustable, capacitors and inductors,
termed the \emph{varcapacitor} and \emph{varinductor}:
\begin{IEEEeqnarray*}{rCll}
        i(t) &=& c(t) \td{}{t}(c(t) v(t))&\qquad\text{varcapacitor}\\
        v(t) &=& l(t) \td{}{t}(l(t)i(t))&\qquad\text{varinductor}.
\end{IEEEeqnarray*}
If $i(t)$, $v(t)$, $l(t)$ and $c(t)$ are $T$-periodic, these devices are monotone on
$L_{2, T}$.

\begin{proposition}
        Varcapacitors with $T$-periodic $c(t)$ and varinductors with $T$-periodic
        $l(t)$ are monotone on $L_{2, T}$. 
\end{proposition}
\begin{proof}
For a varcapacitor, we have
\begin{IEEEeqnarray*}{cl}
        &\int_0^T (v_1(t) - v_2(t))(i_1(t) - i_2(t)) \dd{t} \\
        =& \int_0^T c(t)(v_1(t) - v_2(t))\td{}{t}(  c(t) (v_1(t) - v_2(t))) \dd{t}\\
        =& \int_0^T \td{}{t} \frac{1}{2} c^2(t) (v_1(t) - v_2(t))^2 \dd{t}\\
        =& \frac{1}{2}c^2(T)v^2(T) - \frac{1}{2}c^2(0)v^2(0)\\
        =& 0.
\end{IEEEeqnarray*}
Likewise, for a varinductor, we have
\begin{IEEEeqnarray*}{+cl+x*}
        &\int_0^T (v_1(t) - v_2(t))(i_1(t) - i_2(t)) \dd{t} \\
        =& \int_0^T l(t)(i_1(t) - i_2(t))\td{}{t}(  l(t) (i_1(t) - i_2(t))) \dd{t}\\
        =& \int_0^T \td{}{t} \frac{1}{2} l^2(t) (i_1(t) - i_2(t))^2 \dd{t}\\
        =& \frac{1}{2}l^2(T)i^2(T) - \frac{1}{2}l^2(0)i^2(0)\\
        =& 0.&\qedhere
\end{IEEEeqnarray*}
\end{proof}

We now give two detailed examples of the steady-state analysis of an RLC circuit.
In order to obtain relations on $l_{2, \tau}$,
the derivative is discretized to give an operator $D$.  Any discretization may be
used.  For the examples in this paper, we use the backwards finite difference, given by the relation
\begin{IEEEeqnarray*}{rCl}
        D = \bigg\{ (u, y) \; \bigg| \; y = \tau D_\tau u \bigg\}, 
\end{IEEEeqnarray*}
where  $D_\tau$ is the $\tau \times \tau$ matrix
\begin{IEEEeqnarray*}{rCl}
D_\tau &=& 
\begin{bmatrix}
        1 & 0 &  \ldots & 0  & -1\\
        -1 & 1 &  \ldots & 0 & 0\\
        0 & -1 &  \ldots & 0 & 0\\
        \vdots &  \vdots & \ddots & \vdots &
        \vdots \\
        0 & 0 &  \ldots & -1 & 1
\end{bmatrix}.
\end{IEEEeqnarray*}

Note that $D$ is a maximal monotone relation, as $D_\tau + D_\tau\tran \succeq 0$
\autocite[$\S2.2.3$]{Ryu2021a}.  To obtain an accurate discrete model, a sufficient number of time steps must be used.

\begin{example}\label{ex:envelope_detector}

An envelope detector is a simple nonlinear circuit consisting of a diode in series
with an LTI RC filter
 (Figure~\ref{fig:envelope_detector}).  It is used to demodulate AM radio signals.

\begin{figure}[hb]
        \centering
        \includegraphics{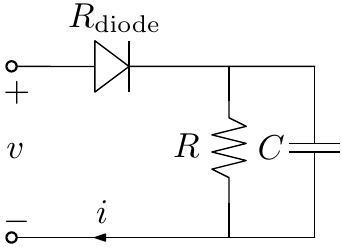}
        \caption{An envelope detector, configured as a 1-port.}%
        \label{fig:envelope_detector}
\end{figure}

We model the
diode using the Shockley equation:
\begin{equation*}
        v_{\text{diode}} = R_{\text{diode}}(i) \coloneqq nV_T \ln\left(\frac{i}{I_s} + 1\right),
\end{equation*}
where $I_s$ is the reverse bias saturation current, $V_T$ is the thermal voltage and $n$
is the ideality factor.  The $i-v$ graph of the diode relation is strictly increasing
with no endpoints; the diode relation is therefore maximal monotone.

The RC filter is itself a parallel interconnection of a resistor and capacitor, which maps voltage to current:
\begin{IEEEeqnarray*}{rCl}
        G_{RC} = C D + \frac{1}{R}I,
\end{IEEEeqnarray*}
where $I$ is the $\tau \times \tau$ identity matrix.
As $G_{RC}$ is linear, it has a Lipschitz constant $L$ equal to its
largest singular value (which will grow with the size of $D$) and is coercive with constant $m = \lambda_{\min}((G_{RC} +
G_{RC}\tran)/2)$ \autocite{Ryu2021a}. 

The incremental voltage $\Delta v = v - v^\star$ is given as a relation of $i$ by
\begin{IEEEeqnarray*}{rCl}
        \Delta v = R_\text{diode}(i) + R_{RC}(i) - v^\star.
\end{IEEEeqnarray*}

Given an input voltage $v^\star$, we solve for the corresponding current $i^\star$ using the Douglas-Rachford splitting.  This
involves applying both the resolvents $\res_{RC}$ and $\res_{\text{diode}}$.
The resolvent $\res_{RC}$ is given by $(I + \lambda G_{RC}^{-1})^{-1}$.  This matrix
is pre-computed and stored in memory.
The resolvent of the diode, $\res_{\text{diode}}$, is given by $\res_{\text{diode}}^{-1}(x) = (I + \lambda
R_\text{diode}(x) - \lambda v^\star)$. There
is no analytic expression for this operator.  Rather, the resolvent is computed
numerically using the guarded Newton algorithm \cite{Parikh2013}.  

Figure~\ref{fig:envelope_detector_inverse} shows the results of performing this scheme with
an input of $v^\star = \sin(2\pi t)$ A, with $R = 1\, \Omega$, $C = 1$ F, $I_s =
1\times10^{-14}$ A, $n = 1$ and $V_T = 0.02585$ V.  The number $\tau$ of time steps used is
$500$.
\end{example}

\begin{figure}[h]
        \centering
        \begin{tikzpicture}
                \begin{groupplot}
                        [
                        group style={
                                group size=1 by 2,
                                vertical sep = 0.5cm
                        },
                        width=0.5\textwidth,
                        height=4cm,
                        cycle list name=colors,
                        grid=both,
                        grid style={line width=.1pt, draw=Gray!20},
                        axis x line=bottom,
                        axis y line=left
                        ]
                        \nextgroupplot[ylabel={\footnotesize Voltage (V)}, xmin=0, xmax=1]
                        \addplot[CornflowerBlue] table [x = t, y = v, col sep = comma, mark = none]{"./envelope_detector_inverse.csv"};
                                \addlegendentry{\footnotesize $v$ - input};
                                
                                \nextgroupplot[xlabel={\footnotesize Time (s)}, ylabel={\footnotesize Current (A)}, xmin=0, xmax=1]
                        \addplot[BurntOrange] table [x = t, y = i, col sep = comma, mark = none]{"./envelope_detector_inverse.csv"};
                                \addlegendentry{\footnotesize $i$ - output};
                                
                \end{groupplot}
        \end{tikzpicture}
        \caption{Input voltage $v^\star$ and the resulting current $i$ for an
        envelope detector.  One period of a periodic input and output is shown. Circuit parameters are $R = 1\, \Omega$, $C = 1$ F, $I_s =
1\times10^{-14}$ A, $n = 1$ and $V_T = 0.02585$ V.
Algorithmic parameters are $\alpha = 0.01$, $\epsilon = 1\times10^{-5}$ and $500$ time
steps.}%
        \label{fig:envelope_detector_inverse}
\end{figure}
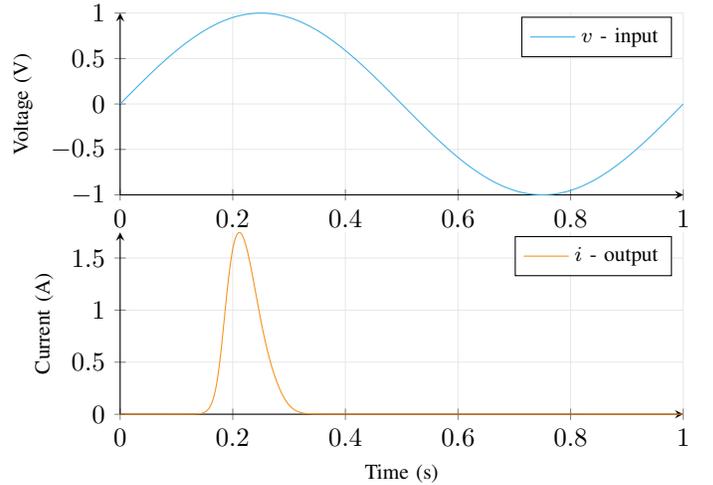

\begin{example}\label{ex:large_scale}
 In this example, we analyze the large-scale circuit shown in 
 Figure~\ref{fig:large_scale}, which consists
 of $n$ identical units, each consisting of a diode and LTI RC filter.  The diode and
 RC filter are modelled as in Example~\ref{ex:envelope_detector}.

\begin{figure*}[t]
        \centering
        \includegraphics{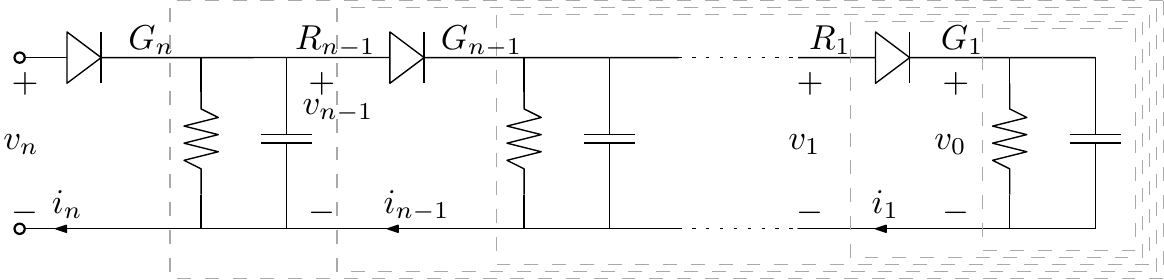}
        \caption{Circuit for Example~\ref{ex:large_scale}.  The circuit consists of
        $n$ repeated units, each consisting of a diode and LTI RC filter.}%
        \label{fig:large_scale}
\end{figure*}

When viewed as an interconnection of one-ports, the circuit has a recursive
structure.  Following the notation of Figure~\ref{fig:large_scale}, for $1 \leq m \leq n$, we have:
\begin{IEEEeqnarray*}{rCl}
        v_{m} &=& R_m(i_m)\\
              &=& R_{\text{diode}}(i_m) + G_{n}^{-1}(i_m),\\
        i_{m} &=& G_{m-1}(v_{m - 1})\\
              &=& G_{RC}(v_{m-1}) + R^{-1}_{m-1}(v_{m-1}).
\end{IEEEeqnarray*}
The base case is $G_1 = G_{RC}$.  This circuit has the form of
Figure~\ref{fig:nested_algo}, with $R_0$ = $G_{RC}^{-1}$, $R_j =  R_{\text{diode}}$ and
$G_j = G_{RC}$ for all $j>1$.  The circuit is solved using the nested
forward/backward algorithm introduced in Section~\ref{sec:nested}.

Figure~\ref{fig:large_scale_result} shows the results of performing this
scheme with $n = 100,000$ repeated units (a total of 300,000 components).  
The input is $v^\star = 1 +\sin(2\pi t)$ A, with circuit parameters $R = 1\, \Omega$, $C = 1$ F, $I_s =
1\times10^{-14}$ A, $n = 1$ and $V_T = 0.02585$ V.  The number of time steps used is
$256$.  Every $\alpha_j$ is set to $1.5$. With $n=100,000$ units, computation took 1937 s on a standard desktop
computer.  With $n=10$ units, computation took an average of 243 ms, over 21 runs.  

In this configuration, the circuit and algorithm do not satisfy the convergence
criteria of Theorem~\ref{thm:nested_convergence}.  $R_0$ is $1$-Lipschitz and
$1/512$-coercive (verified by calculating its singular values and spectrum),
so the theorem applies, however the $\alpha_j$ cannot be selected to make the
spectral radius of $\mathcal{A}$ less than $1$.  If, however, resistors are
placed in parallel with each diode, these elements become coercive, and the
$\alpha_j$ can be chosen so that convergence is guaranteed by
Theorem~\ref{thm:nested_convergence}.  The range of allowable step sizes is made
larger by inserting a resistor in series with the capacitor in $R_0$ (which increases
its coercivity constant).  The insertion of fictitious resistances to aid
convergence is standard in circuit simulation -- see, for instance, the
$\texttt{gmin}$ method in SPICE \autocite[$\S$15.3.5]{Vogt2022}.  The fact that
this example converges rapidly \emph{without} employing this trick further motivates
the search for less conservative convergence conditions than that given by
Theorem~\ref{thm:nested_convergence}.
\end{example}

\begin{figure}[h]
        \centering
        \begin{tikzpicture}
                \begin{groupplot}
                        [
                        group style={
                                group size=1 by 2,
                                vertical sep = 0.5cm
                        },
                        width=0.5\textwidth,
                        height=4cm,
                        cycle list name=colors,
                        grid=both,
                        grid style={line width=.1pt, draw=Gray!20},
                        axis x line=bottom,
                        axis y line=left
                        ]
                        \nextgroupplot[ylabel={\footnotesize Voltage (V)}, xmin=0, xmax=1]
                        \addplot[CornflowerBlue] table [x = t, y = v, col sep =
                                comma, mark = none]{"./large_scale_100k.csv"};
                                \addlegendentry{\footnotesize $v_n$ - input};
                                
                                \nextgroupplot[xlabel={\footnotesize Time (s)}, ylabel={\footnotesize Current (A)}, xmin=0, xmax=1]
                        \addplot[BurntOrange] table [x = t, y = i, col sep = comma,
                                mark = none]{"./large_scale_100k.csv"};
                                \addlegendentry{\footnotesize $i_n$ - output};
                                
                \end{groupplot}
        \end{tikzpicture}
        \caption{Input voltage $v_n$ and the resulting current $i_n$ for the
        circuit of Figure~\ref{fig:large_scale}, with $n=100,000$. Circuit parameters
        are $R = 1\, \Omega$, $C = 1$ F, $I_s =
1\times10^{-14}$ A, $n = 1$ and $V_T = 0.02585$ V.  Algorithm parameters are
$\alpha_j = 1.5$ for all $j$ and $\epsilon = 1\times10^{-4}$.  One period of a periodic input and output is shown.}%
\label{fig:large_scale_result}
\end{figure}
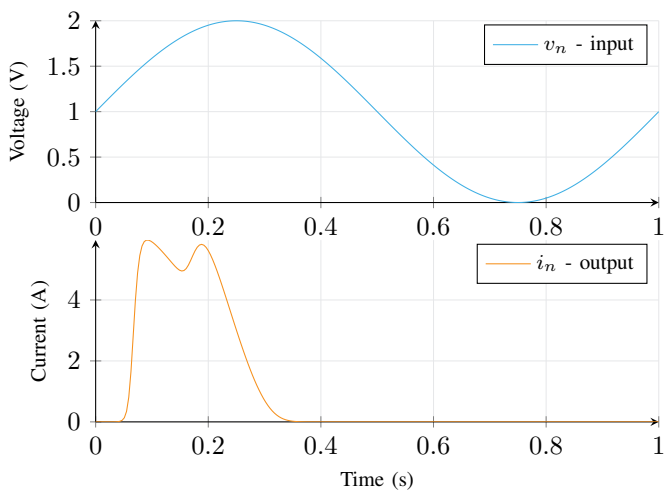

\section{Memristive circuits}
\label{sec:conductance}

In 1976, \textcite{Chua1976} introduced the class of 
\emph{memristive systems}, described by state space models of the form
\begin{subequations}\label{eq:memristive}
\begin{IEEEeqnarray}{rCl}
        \dot{x} &=& f(x, i, t)\\
              v &=& g(x, i, t)i.
\end{IEEEeqnarray}
\end{subequations}
This model class describes systems which behave like resistors, in that they
cannot store energy and do not produce a phase shift, but, unlike resistors, do have memory.  This
work was motivated by systems such as the Hodgkin-Huxley neural membrane model
\autocite{Hodgkin1952}, thermistors and discharge tubes.
In this section, we investigate the maximal monotonicity of some memristive circuits.

If $\dot{x} = f(x, u)$ is a contractive, time-invariant state space system and $u(t)$ is a
$T$-periodic input, there is a unique, globally asymptotically stable $T$-periodic
output $y(t)$ to the memristive system~\eqref{eq:memristive} \autocite{Lohmiller1998}.  The
memristive system then defines an operator on $L_{2, T}$, mapping the $T$-periodic
input $u(t)$ to the $T$-periodic output $y(t)$.

To determine the monotonicity properties of memristive systems, 
we use the Scaled Relative Graph (SRG).
The SRG of an operator is a region of the extended complex plane, from which the
incremental properties of the operator can be easily read.  SRGs were recently
introduced by \textcite{Ryu2021} for the study of monotone operator methods in
optimization, and have been used for the study of systems in feedback by the authors
in references \autocite{Chaffey2021c, Chaffey2022a}.  We refer the interested reader
to these references for the details of SRG analysis.
An operator is
$\mu$-monotone if and only if its SRG lies in the region $\{z \in \C\;|\; \Re(z) \geq
\mu\}$ \autocite[Prop. 3.3 and Thm. 3.5]{Ryu2021}.

\begin{example}\label{ex:conductance}
The Hodgkin-Huxley model represents a nerve axon membrane as a parallel
interconnection of active \emph{ion channels} with a capacitor
\autocite{Hodgkin1952}.  Each ion channel is a time-varying conductance, which may be
modelled as a memristive system.  In this example, we consider
the potassium conductance $i = G_\mathrm{K}(v)$, which is given by the equations
\begin{IEEEeqnarray*}{rCl}
        i &=& \bar{g}_\mathrm{K} n^4(v - v_\mathrm{K})\\
        \td{n}{t} &=& \alpha_n (v)(1 - n) - \beta_n(v)n\\
        \alpha_n(v) &=& \frac{0.01 (10 + v)}{\exp(1 + v/10) - 1}\\
        \beta_n(v) &=& 0.125\exp(v/80).
\end{IEEEeqnarray*}
Following \textcite{Hodgkin1952a}, the constants $\bar{g}_\mathrm{K}$ and
$v_\mathrm{K}$ are set to $19$ m mho$/$cm$^2$ and $12$ mV, respectively.
The dynamics in $n$ are contractive \autocite[Prop. 1]{Burghi2021}, therefore the
potassium conductance defines an operator on $L_{2, T}$.

The analytical SRG of the potassium conductance is difficult to determine, but 
we can test its monotonicity by sampling its SRG.
Figure~\ref{fig:K_srg} shows points in the SRG of the potassium conductance, computed
over signals of the form $u = \alpha \sin(\gamma t) + \delta$, for real parameters
$\alpha, \gamma, \delta$.  This plot suggests that the potassium conductance is $(-0.002)$-monotone on $L_{2, T}$.  
While we do not have a theoretical guarantee that this is the case, we can test
whether the potassium conductance behaves as if it is $(-0.002)$-monotone when
connected in a circuit. 

\begin{figure}[hb]
        \centering
        \includegraphics{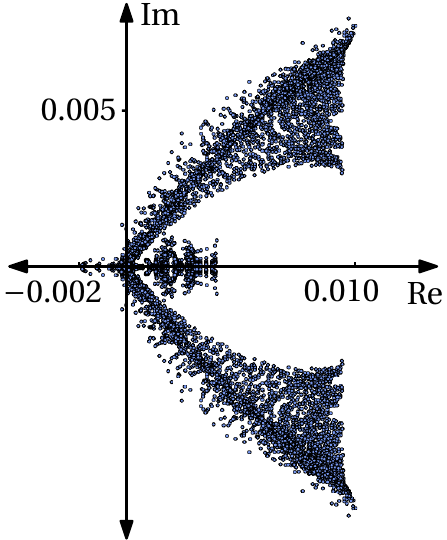}
        \caption{Sampling of the SRG of a potassium conductance. Each point $z$
        corresponds to a pair of input currents, $i_1$ and $i_2$, and corresponding
voltages $v_1$ and $v_2$.  The magnitude of $z$ is the incremental gain $\norm{v_1 -
v_2}/\norm{i_1 - i_2}$ and the argument of $z$ is the incremental rotation
$\arccos(\ip{v_1 - v_2}{i_1 - i_2}/(\norm{v_1 -
v_2}\norm{i_1 - i_2}))$.}
        \label{fig:K_srg}
\end{figure}
We consider the parallel interconnection of the potassium current with an LTI resistor.  The port relation of this circuit is given by
\begin{IEEEeqnarray*}{rCl}
        i = R^{-1} v + G_\mathrm{K}(v).
\end{IEEEeqnarray*}
Given a periodic input current $i^\star$, the corresponding voltage is solved using
the forward/backward algorithm.  The forward step is applied to $G_\mathrm{K}$, and
the backward step is applied to $R^{-1}$.  The algorithm
converges when $R \geq 500\, \Omega$, supporting the hypothesis that the potassium
conductance is $(-0.002)$-monotone. The Lissajous figure, or $i-v$ plot, is shown in
Figure~\ref{fig:K_lissajous} for an input current of $i(t) = \sin(2\pi t)$.  The
potassium conductance exhibits the characteristic zero-crossing Lissajous figure of a
memristive system \autocite{Chua1976}.
\end{example}

\begin{figure}[h]
        \centering
        \begin{tikzpicture}
                \begin{axis}
                        [
                        no markers,
                        name = ax1,
                        width=0.4\textwidth,
                        height=0.35\textwidth,
                        ticklabel style={/pgf/number format/fixed},
                        ylabel={\footnotesize Current (A)},
                        xlabel={\footnotesize Voltage (V)},
                        cycle list name=colors,
                        grid=both,
                        grid style={line width=.1pt, draw=Gray!20},
                        axis x line=bottom,
                        axis y line=left
                        ]
                        \addplot table [x=v, y=i, col sep = comma, mark =
                                none]{"./K_current.csv"};
               \end{axis}
        \end{tikzpicture}
        \caption{Lissajous figure of a potassium conductance in parallel with a $500\;
                \Omega$ resistor.  The large signal magnitudes are not physically
        realistic, but are chosen for illustrative purposes.  The Lissajous figure
always passes through the origin, a fundamental characteristic of a memristive
system.}%
        \label{fig:K_lissajous}
\end{figure}
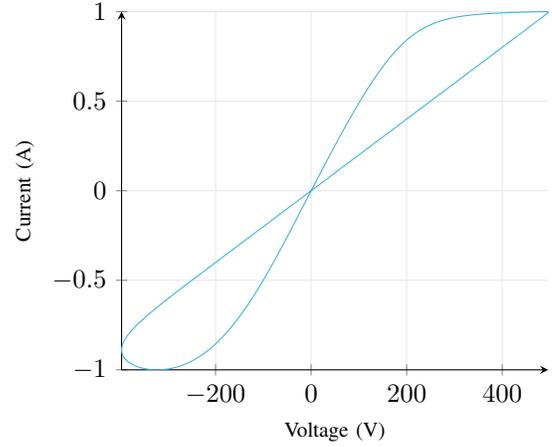

\section{Conclusions}
\label{sec:conclusions}

We have shown that monotone -- unlike passive -- nonlinear circuits retain the strong
physical and computational significance of passive linear circuits. Splitting
algorithms allow the computation to be separated in a way which mirrors the topology
of the circuit, and a new splitting algorithm has been introduced which is suited to
circuits with nested series/parallel interconnections.  This method has been demonstrated on 
the classes of circuits built from maximal monotone
resistors and LTI capacitors and inductors, and memristive dynamic conductances such
as the neuronal conductances of the Hodgkin-Huxley model.

The mathematical property of monotonicity connects the physical property of energy
dissipation with a well-established algorithmic theory for computation, for systems
modelled as nonlinear operators.  This mirrors the connection between energy
dissipation in LTI state space systems and computational methods for LMIs,
established by the theory of dissipativity \autocite{Willems1972}.  Preliminary work
by the authors \autocite{Das2022} shows that the algorithmic methods proposed here
may be extended beyond the class of systems formed by the interconnection of monotone
elements, to those systems formed by the \emph{difference} of monotone elements.
This includes systems with self-sustaining oscillations.

\section{Acknowledgements}

The authors gratefully acknowledge many insightful discussions with Fulvio Forni,
whose suggestions greatly improved this manuscript.  We would also like to thank the
anonymous reviewers for their useful feedback and suggestions.

\printbibliography

\begin{IEEEbiography}[{\includegraphics[width=1in,height=1.25in,clip,keepaspectratio]{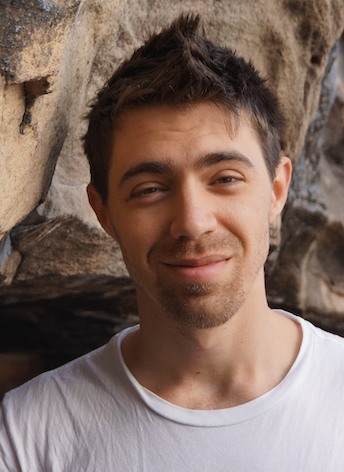}}]{Thomas Chaffey} (M17) received the B.Sc (advmath) degree in
        mathematics and computer science in 2015 and the M.P.E degree in mechanical
        engineering in 2018, from the University of Sydney, Australia.  He received
        the Ph.D. degree in 2022 from the University of Cambridge, and currently
        holds the Maudslay-Butler Research Fellowship at Pembroke College, Cambridge. 
        His research interests are in nonlinear control, optimization and
        circuit theory.  He received the Best Student Paper Award at
        the 2021 European Control Conference and the Outstanding Student Paper Award
        at the 2021 IEEE Conference on Decision and Control.
\end{IEEEbiography}

\begin{IEEEbiography}[{\includegraphics[width=1in,height=1.25in,clip,keepaspectratio]{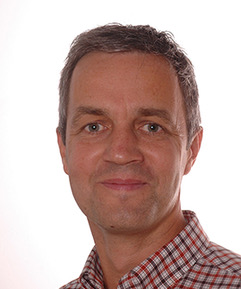}}]{Rodolphe Sepulchre} (M96,SM08,F10)  received the engineering degree and the Ph.D. degree from the Université catholique de Louvain in 1990 and in 1994, respectively.  His is Professor of engineering at Cambridge University since 2013.  His research interests are in nonlinear control and optimization, and more recently neuromorphic control.  He co-authored the monographs "Constructive Nonlinear Control" (Springer-Verlag, 1997) and "Optimization on Matrix Manifolds" (Princeton University Press, 2008). 
He is Editor-in-Chief of IEEE Control Systems. In 2008, he was awarded the IEEE Control Systems Society Antonio Ruberti Young Researcher Prize. He is a fellow of IEEE, IFAC,  and SIAM. He has been IEEE CSS distinguished lecturer  between 2010 and 2015. In 2013, he was elected at the Royal Academy of Belgium.
\end{IEEEbiography}

\end{document}